\documentclass[journal]{elsarticle} % ``review'' or ``journal'' or ``p5'' (two columns)

\usepackage{hyperref}
\usepackage{amsmath}
\usepackage{amssymb}
\usepackage{color}

\journal{---}

\newtheorem{theorem}{Theorem}[section]

\newenvironment{definition}[1][Definition]{\begin{trivlist}
\item[\hskip \labelsep {\bfseries #1}]}{\end{trivlist}}
\newenvironment{proof}[1][Proof]{\begin{trivlist}
\item[\hskip \labelsep {\bfseries #1}]}{\end{trivlist}}

\begin{document}

\begin{frontmatter}

\title{Mathematical transcription of the ``Time-Based Resource Sharing'' theory of working memory}

\author{Nicolas Gauvrit}
\address{Human and Artificial Cognition Lab, \'Ecole Pratique des Hautes \'Etudes, Paris}
\fntext[myfootnote]{\texttt{ngauvrit@me.com}}

\author{Fabien Mathy}
\address{Universit\'e Nice Sophia-Antipolis}




\begin{abstract}
The time-based resource sharing (TBRS) model is a prominent model of working memory that is both predictive and simple. The TBRS is the mainstream decay-based model and the most susceptible to competition with interference-based models. A connectionist implementation of the TBRS, the TBRS*, has recently been developed. However, the TBRS* is an enriched version of the TBRS, making it difficult to test the general characteristics resulting from the TBRS assumptions. Here, we describe a novel model, the TBRS2, built to be more transparent and simple than the TBRS*. The TBRS2 is minimalist and allows only a few parameters. It is a straightforward mathematical transcription of the TBRS that focuses exclusively on the activation level of  memory items as a function of time. Its simplicity makes it possible to derive several theorems  from the original TBRS and allows several variants of the refreshing process to be tested without relying on particular architectures.
\end{abstract}


\end{frontmatter}

% \linenumbers
%%%%%%%%%%%%%%
%%%%%%%%%%%%%%

Working memory is often described as a unique cognitive resource serving both short-term maintenance and processing \cite{baddeley2003working}. It is thus central in problem solving \cite{swanson2004relationship}.
Working memory is known to be mediated by the prefrontal cortex \cite{braver2002role,kane2002role,prabhakaran2000integration}, and it is clearly linked to intelligence \cite{conway2002latent,engle1999working},  particularly when complex span tasks are used \cite{unsworth2009complex}.

The \emph{complex span task} paradigm is probably the most widespread dual procedure for measuring working memory \cite{conway2005working}. In complex span tasks, participants are presented with items to be memorized and instructed to recall them at the end of the trial in the correct order. In contrast to simple span tasks, the  presentation of the memory items is interspersed with a concurrent task that presents distractor items. Participants thus alternate between encoding memoranda and processing distractors.

A variety of concurrent tasks have been used, including reading sentences \cite{van2008standard}, reading digits \cite{barrouillet2004time}, verifying arithmetic statements \cite{redick2015measuring,turner1989working},  uttering predetermined syllables \cite{lewandowsky2010turning}, and diverse elementary tasks involving executive functions such as memory retrieval, response selection, or updating \cite{barrouillet2011law}. A critical feature for determining the validity of complex span tasks is the degree of control the experimenter has over the time devoted to processing the distractors vs. encoding the memory items \cite{oberauer2011tbrs}. In this respect, a computerized version of a complex span task \cite{barrouillet2007time} offers fine control over a participant's processing timeline.

In a computerized version of the complex span task, participants are, for instance, required to memorize a series of  items appearing on a screen sequentially. Each item is displayed once for a fixed amount of time (often fixed  to 1 s). Between these memory items, participants have to perform a second task. For instance, for the second task a participant could see a ``2'' for a short period of time, then ``+3,'' then ``$-$1'' and be required to update the digit 2 into 5 and then 4. At each step, the participant is instructed to give out loud the current result of the series of arithmetic operations. At the end of the trial, the participant must recall the to-be-memorized items in order at their own pace. A  trial for a list of two letters would look like: L, 2, +3, $-$1, H, 3, +1, $-$2, Recall. A correct response would be ``LH'' in that case.

Barouillet, Bernardin, and Camos \cite{barrouillet2004time} developed the
\emph{time-based resource sharing} (TBRS) model to specifically account for the performance of participants  in such experiments. The key ideas of the TBRS model can be summarized as follows: (1) Unless attention is focused on the memorandum to refresh the memory items, their memory traces fade away. (2) Because of a bottleneck effect, attention can be devoted only to either refreshing the memorandum item by item or  processing the distractors. Thus, participants perform rapid switches between refreshing items and performing the concurrent task. (3) The probability that an item is correctly recalled is a function of the item \emph{cognitive load}, defined as the proportion of time that cannot be devoted to refreshment of the item.

The TBRS has received much qualitative empirical support \cite{barrouillet2012time,barrouillet2009working,portrat2008time,vergauwe2009visual}. However, although it is both simple and rooted in a set of clearly stated hypotheses, the model remains underspecified. First, it does not indicate how long the refreshment period (between two switches) lasts for an item and whether this duration is fixed or determined by specific factors. It also does not indicate what happens when refreshment has been interrupted by a distractor, that is, whether people start refreshing the first item anew or continue from the last refreshed item, and so forth. Second, although the TBRS is described in detail in several publications \cite{barrouillet2004time, barrouillet2007time}, it has remained a theory based mostly on a verbal description (until recently---see below). This is unfortunate because it reduces opportunities to test the model and to use statistical criteria of fit \cite{pitt2002good}. Moreover, because building a mathematical description of the TBRS is easy and straightforward, verbal descriptions should be avoided by all means \cite{farrell2010computational,pothos2011formal,norris2005computational}.

Oberauer and Lewandowsky \cite{oberauer2011tbrs} have developed the only available computational implementation of the TBRS, the \emph{TBRS*}, a two-layer connectionist network. One objective of the authors was to bridge existing computational models of working memory with prominent features of the TBRS model. Although the model was found to fit experimental data, the authors remain skeptical about the TBRS in their conclusion.

The TBRS* is an important step toward a precise quantitative validation of the TBRS, but it may be a model that is too enriched compared to the original TBRS. Some caveats should therefore be kept in mind. First, the TBRS* merges features from the TBRS and features from other models. Hence, if it fails at predicting empirical data, it would be unclear whether this failure should be taken as evidence against the TBRS or against some other features related to the specific implementation. 
Second, because the TBRS* is an enriched model, some decisions were made (e.g., a serial position coding) that would be unnecessary for defining a more basic implementation of the TBRS.
Last, connectionist networks can be seen as ``black boxes'' and may lack transparency \cite{shahin2009recent}, making it difficult to formally demonstrate results that nonetheless follow from the TBRS framework.

Here, we present a new mathematical implementation of the TBRS (henceforth TBRS2) designed to remain as close as possible to the TBRS assumptions. On the one hand, our implementation is not as rich as the TBRS* and cannot account for as many features, as we do not address any question not already addressed in the TBRS verbal description. For instance, we do not aim to code the order of items in any way but stick to predicting the correct recall of each item. On the other hand, the TBRS2 bears two interesting features: (1) We have to make a decision only about the decay function of memory traces and the schedule of refreshing, that is, 2 decisions, whereas the TBRS* has to make 11 such decisions. In the same vein, our model needs 4 parameters, whereas the TBRS* needs no less than 10. (2) The TBRS2 dynamics is more transparent than that of the TBRS* because it relies on a mere analytical translation of the TBRS assumptions instead of a connectionist model. As a result, we can mathematically \emph{prove} some consequences of the TBRS assumptions, such as a functional entanglement between the decay and refreshing functions.

In the first three sections, we will describe TBRS2 and prove some theorems directly following from the TBRS assumptions. Six variants of the TBRS2 will then be described that vary depending on how the refreshment process occurs. In the last section, we will use empirical data to illustrate how the TBRS2 can be used to test the TBRS quantitatively and precisely.

%%%%%%%%%%%%%%
%%%%%%%%%%%%%%
\section{Overview}

Figure \ref{fig:flowchart} provides a general overview of the TBRS2. First, a complex span task (as described above) can be modeled by a ``task function,'' that is, a function of time indicating whether a memorandum is presented, a processing task is performed, or neither of these events occur, at time $t$. When an item is presented, attention focuses on the item (according to the TBRS). When a processing task is performed, the attentional focus is driven away from memoranda. The remaining ``spare'' time is dedicated entirely to refreshing items. How refreshment of the items is spread along the timeline depends on a refreshment strategy that the TBRS left unspecified. Once a decision is made about the refreshment schedule, we can derive a ``focus function'' from the task function. The focus function is a function of time indicating attentional focus (toward processing or encoding/refreshing an item). The graphical example given in Figure \ref{fig:flowchart} is based on the assumption that each item is refreshed for a fixed duration, starting anew from the first item after each interruption.

\begin{figure}[htbp]
\begin{center}
\includegraphics[width=\textwidth]{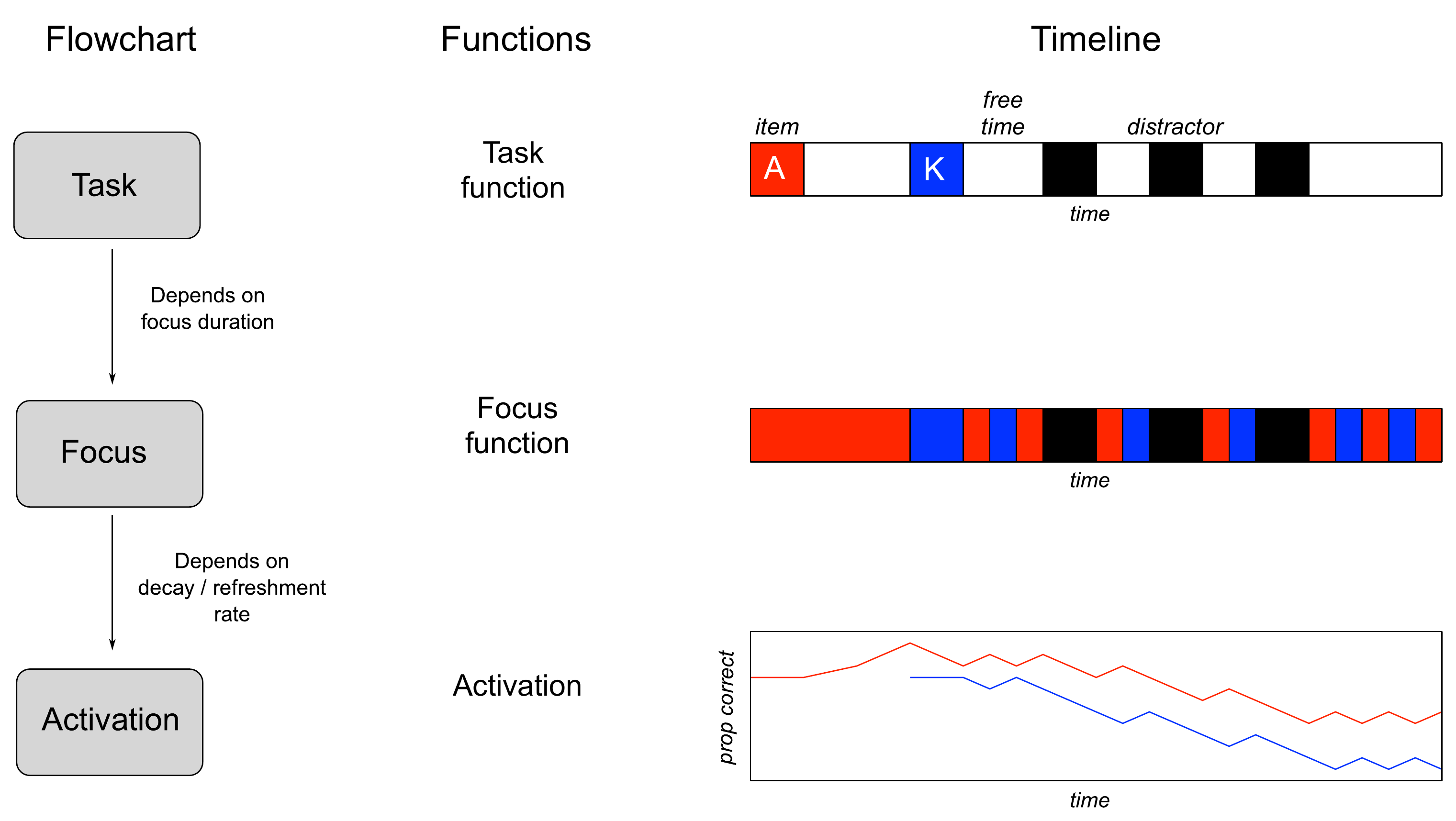}
\caption{Overview of TBRS2's main functions. Complex span tasks are modeled by a task function indicating what is happening along the timeline [presentation of a new memorandum (in color), free time (white), or processing tasks (black)]. From the task function, we can derive a focus function indicating how attentional focus switches from items (color) to processing distractors (black). To translate task into focus, we must specify how spare time is used, e.g., how long each refreshment period lasts. Activation is the odds of correct recall for a given item and can be derived from the focus function as soon as the decay and refreshment rates are set.}
\label{fig:flowchart}
\end{center}
\end{figure}

Once the focus function is set, we can derive the dynamics of activation of each item. Here, \emph{activation} stands for the odds of correct recall of an item at a given time, that is, the odds that the item would be recalled if the participants had to recall it at time $t$. The focus function translates into the activation dynamics through the decay and refreshment functions: when attentional focus spots item $i$, the activation of $i$ increases by the effect of refreshment. Otherwise, the activation decreases by the effect of time. As we will demonstrate below, the TBRS assumptions impose a direct link between the decay and refreshment functions. In the following, we will give more details about the model, starting from the link between the focus function and activation, and then turning to the task functions and their relations with the focus functions.

%%%%%%%%%%%%%%
%%%%%%%%%%%%%%
\section{From focus function to activation}

Let us consider a situation in which a sequence of items $(x_{1},\ldots ,x_{k})$ is to be memorized within a period of time $[0,T]$. At any instant $t\in [0,T]$ after its presentation, item $x_{i}$ has an \emph{activation} $a_{i}(t)$, here defined as the odds that $x_{i}$ would be correctly recalled at time $t$ by the participant, who would be required to recall items at this point in time.

%%%%%%%%%%%%%%
\subsection{The dynamics of activation}

We later assume that the decay of memory traces (i.e., activation) is exponential, following a previous study \cite{wickelgren1970time}. However, we first consider a more general case of decay in order to prove a general link between decay and refreshment rates that follows from the TBRS hypotheses.
When attention is focused toward item $x_{i}$, the corresponding activation increases in such a way that $a_{i}$ is a solution of an autonomous differential equation $y'=R(y).$ The refreshment function $R$ is  continuous positive and does not depend on $i$. The previous equation only formalizes the idea that the refreshment rate does not depend on time per se, but depends on the current activation of the item. Likewise, when attention is driven away from $x_{i}$, $a_{i}$ decreases following an autonomous differential equation $y'=-D(y)$, where $D$ is a continuous positive function that is independent of $i$.

%%%%%%%%%%%%%%
\subsection{Focus function and cognitive load}

Let us first define a function describing the dynamics of attentional focus with respect to an item $x_{i}$.

\begin{definition}
The \emph{focus function} $\varphi_{i}$ of item $x_{i}$ is defined as $\varphi_{i}(t)=1$ if the item is being refreshed at time $t$ (i.e., attention is focused on $x_{i}$), $\varphi_{i}(t)=2$ if item $x_{i}$ is displayed at time $t$, and $\varphi_{i}(t)=0$ otherwise\footnote{Note that $\varphi_{i}(t)=2$ is only a dummy code ascribed to a particular event.}.
\end{definition}

\begin{definition}
The \emph{item cognitive load} associated with $x_{i}$ on a time interval $[t_{0},t_{1}]$ is defined as the proportion of time not devoted to the item, that is, 
$$CL_{i}(t_{0},t_{1})=\frac{\mu \{[t_{0},t_{1}]\cap \varphi^{-1}(0) \}}{t_{1}-t_{0}}.$$
\paragraph{Note} When computing this cognitive load, we will always assume that item $x_{i}$ was presented before $t_{0}$ (and thus does not appear during $[t_{0},t_{1}]$).
\end{definition}

The TBRS model's core assumption is that $a_{i}(t_{1})$ depends only on $CL_{i}(t_{0},t_{1})$ and on its initial value $a_{i}(t_{0})$. From this \emph{cognitive load assumption}, a relation between $D$ and $R$ can be derived:

\begin{theorem}
\label{theo:cl}
Under the cognitive load assumption, $R$ and $D$ are proportional: $R=\kappa D,\; \kappa\in \mathbb{R}_{+}^{*}$.
\end{theorem}

\begin{proof}
Consider a situation in which a unique item ($x$) is to be remembered at time $T$ and such that attention is driven away from $x$ except on an interval $[\tau,\tau+h]$ (think of $h$ as ``small'').

The activation $a(t)$ is thus decreasing on $[0,\tau]$ and $[\tau+h,T]$, but increasing on $[\tau, \tau+h]$. The cognitive load at  $T$  does not depend on $\tau \in [0,T-h]$, so $a(T)$ does not depend on $\tau$ either.

There exists a single $\epsilon>0$ such that $a(\tau+h+\epsilon)=a(\tau)$, where $\epsilon$ is the time needed for $a$ to go back down to the level it was at $\tau$, before refreshment. The cognitive load assumption implies that $h+\epsilon$ is independent of either $\tau$ or $y=a(\tau)$.

When $h\longrightarrow 0$, so does $\epsilon$. Let $\delta=a(\tau+h)-a(\tau)$. Considering $a(t)$ on $[\tau,\tau+h]$, we find
$$\frac{\delta}{h}\longrightarrow R(y).$$

Considering $a(t)$ on $[\tau+h,\tau+h+\epsilon]$, we have
$$\frac{\delta}{\epsilon}\longrightarrow D(y);$$

thus,
$$\frac{\epsilon}{h}\longrightarrow \frac{R(y)}{D(y)}.$$

Because $h$ and $\epsilon$ are independent of $y$ if the cognitive load assumption is satisfied, we must have $R=\kappa D$, $\kappa=\lim(\epsilon / h)\in \mathbb{R}_{+}^{*}$. $\square$
\end{proof}

Thus, $D$ and $R$ are proportional under the cognitive load assumption expressed in the TBRS model. This is a mathematical consequence of a main TBRS hypothesis that has never been expressed before.

%%%%%%%%%%%%%%
\subsection{Exponential decay}

For the sake of simplicity, we suppose that whenever an item is first presented, its activation equals a constant baseline value $\beta$ during presentation (this does not impair the generality of TBRS2, providing that the presentation duration is constant across memoranda).

Henceforth, we will also assume that the decrease in activation is exponential, which amounts to saying that $D$ is a linear function of $y$. From Theorem \ref{theo:cl}, we know that $R$ is then also a linear function of $y$ (i.e., refreshment is exponential). In other words, if attention is not focused on item $x_{i}$, then $a_{i}(t)\propto \exp (- d t)$, where $d$ is the (absolute) \emph{decay rate}. If attention is focused on $x_{i}$, then $a_{i}(t)\propto \exp(r t)$, where $r$ is the \emph{refreshment rate}.
For exponential decay, an easy way to study the probability of a correct recall is to consider log-odds instead of activation levels (odds). Indeed, if activation decay is exponential, then log-odds evolution is linear, with slope $r$ and $-d$; hence, the following theorem (the proof is immediate):

\begin{theorem}
\label{theo:formula}
Suppose that at time $t_{0}$, an item $x_{1}$ has activation $a_{1}(t_{0})$. Let $\varphi_{1}$ be its focus function. If the item is never presented during period $[t_{0},t]$, then
$$\log (a_{1}(t))=\log(a_{1}(t_{0}))-d (t-t_{0})+(d+r)\int_{t_{0}}^{t}\varphi_{1}(u) du .$$
\end{theorem}

Figure \ref{fig:ex1} displays two simple examples of activation dynamics. The plots were built using the \emph{tbrs} \texttt{R}-function\footnote{Available at \texttt{https://github.com/ngauvrit/tbrs.git}} \cite{team2014r}. Alternatively, one can use our user-friendly online \texttt{Shiny} application\footnote{\texttt{https://mathematicalpsychology.shinyapps.io/tbrs}}.

\begin{figure}[htbp]
\begin{center}
\includegraphics[width=\textwidth]{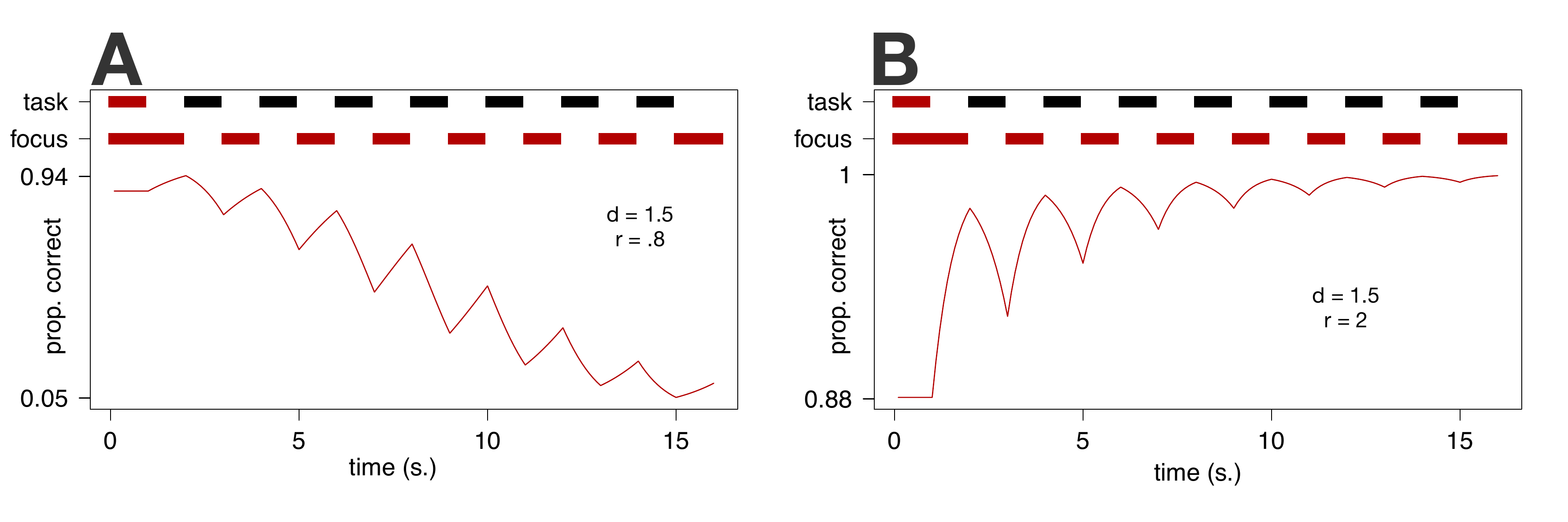}
\caption{Examples of activation dynamics predicted by the model, with a focus on attention switching on and off of the memory item every second. During the first second, the item is being presented and the log-activation is set to 2. Two different sets of conditions $(d,r)$ are presented. The memory trace fades away when $d<r$ (subplot A), but not when $d>r$ (subplot B).}
\label{fig:ex1}
\end{center}
\end{figure}

%%%%%%%%%%%%%%
%%%%%%%%%%%%%%
\section{Task function and task cognitive load}

We have so far considered the case of a single memorandum, but the TBRS was designed to predict item recall in more complex tasks, in which several items $(x_{1},\ldots , x_{n})$ are to be remembered.

\begin{definition}
Define a \emph{task function} as a function $T$ of time, with $T(t)=-1$ if the concurrent task is being performed at $t$, $T(t)=0$ if no concurrent task is to be performed at $t$ with no item presented, and $T(t)=k$, where $k\in \{1,\ldots,n\}$, if item $x_{k}$ is being presented at time $t$.
\end{definition}

\begin{definition}
For such a task, the \emph{cognitive load}
is the proportion of time exclusively devoted to the concurrent task on a given interval:

$$CL(t_{0},t_{1})=\frac{\mu \{ [t_{0},t_{1}]\cap T^{-1}(-1) \} }{t_{1}-t_{0}},$$

\noindent where $\mu$ is the usual Borel measure.
\end{definition}

%%%%%%%%%%%%%%
\subsection{An invariant}

Consider a task function $T$ and a time interval $[t_{0},t_{1}]$ on which no item is presented. Part of this time [$CL(t_{0},t_{1})$] is dedicated to the concurrent task, but the rest is devoted to refreshment. Because refreshment may be distributed among the different items in various specific time courses, the activation at time $t_{1}$ might vary. Consider, for instance, the case of two items $x$ and $y$. If refreshment is dedicated mainly to $x$ during spare time, then we expect a high probability of recall for $x$ and a low one for $y$; however, the reverse case is to be expected if refreshment is dedicated mainly to $y$.
Thus, how ``spare time'' is distributed among the items is an important question for making quantitative predictions (see section \ref{sec:variants}). However, because every spare time is dedicated to a single item, some function of item activations is invariant, as shown by the next theorem.

\begin{theorem}
\label{theo:invariant}
Suppose that no new item is presented on an interval $[t_{0},t_{1}]$, and let $n$ be the number of to-be-remembered items at time $t_{0}$. Then,
$$\prod_{i=1}^{n}a_{i}(t_{1})$$ does not depend on how the refreshment time is distributed on $[t_{0},t_{1}]$. It depends only on $\prod a_{i}(t_{0})$ and $CL(t_{0},t_{1})$.
\end{theorem}
\begin{proof}
Consider $\log (a_{i})$. On any interval dedicated to a dual task, $a_{i}$ is decreasing, with slope $-d$, for all $i$. Thus, the sum
$$S(t)=\sum_{i=1}^{n}\log(a_{i}(t))$$
is a linear function with slope $-nd$.
On any interval on which attention is focused on an item, $S(t)$ is also linear, with slope $r-(n-1)d$.
Therefore, we have
$$S(t_{1})-S(t_{0})=-nd(t_{1}-t_{0})CL(t_{0},t_{1})+(r-(n-1)d)(t_{1}-t_{0})(1-CL(t_{0},t_{1})).$$
Considering the exponential completes the proof. $\square$
\end{proof}

%%%%%%%%%%%%%%
\subsection{Simple span}

The parameters $r$ and $d$ are directly related to the simple span (memory capacity), defined as the maximum number of items one can maintain in memory when no concurrent task is involved. More precisely, the theoretical simple span $k$ can be computed using a simple formula, as shown by Theorem \ref{theo:span}. Using this formula, we can estimate a participant's simple span from any data gathered through, e.g., a complex span task, and thus again test the TBRS assumptions.

\begin{theorem}
\label{theo:span}
Let $k$ be the simple span (memory capacity) corresponding to a set of parameters. We have
$$k=\left\lfloor 1+\frac{r}{d} \right\rfloor .$$
\end{theorem}
\begin{proof}
Consider a simple span task, in which $n$ items are presented sequentially beginning at $t=0$, and involving no dual task.
Let $t_{0}$ be a time at which the $n$ items have been presented. Then, $CL(t_{0},t)$ remains null, and thus
$$S(t)=S(t_{0})+(r-(n-1)d)(t-t_{0}),$$
where $$S(t)=\sum_{i=1}^{n}\log(a_{i}(t)).$$
Thus, $S(t)$ tends to $\pm \infty$, depending on the sign of $r-(n-1)d$.

If $r-(n-1)d>0$, or $n<1+\frac{r}{d}$, then $S(t)$ tends toward $\infty$, and so does $$\prod a_{i}(t)=\exp(S(t)).$$
If $n>1+\frac{r}{d}$, then $S$ tends toward $-\infty$, and $$\prod a_{i}(t)\longrightarrow 0.$$

We have thus proven that

$$\prod_{i=1}^{n}a_{i}(t)\xrightarrow[t\rightarrow \infty]{} 0$$
if $n>1+\frac{r}{d}$, and
$$\prod_{i=1}^{n}a_{i}(t)\xrightarrow[t\rightarrow \infty]{} \infty$$
if $n<1+\frac{r}{d}$.

If $n$ items can be maintained in memory, then no activation tends toward 0, and thus $\prod_{i=1}^{n}a_{i}(t)$ does not tend toward 0, which means that $n\leq 1+r/d$. Thus, $k=\lfloor 1+\frac{r}{d}\rfloor$ is the maximum number of maintainable items, i.e., the simple span. $\square$
\end{proof}

%%%%%%%%%%%%%%
\subsection{Summary}

We implemented the assumptions expressed in the TBRS model in a mathematical framework based on the following axioms:

\begin{enumerate}
\item Attention is always focused on a single to-be-remembered item or on a concurrent task.
\item Activation of item $x_{i}$ increases when attention is focused toward $x_{i}$ and decreases otherwise. The rate of decay/increase is a function of the current activation.
\item Activation of item $x_{i}$ at time $t$ is a function of the cognitive load on $[t_{0},t]$, providing that item $x_{i}$ is presented before time $t_{0}$.
\item Activation decreases exponentially.
\end{enumerate}

From these axioms borrowed from the TBRS (except for the exponential decay), we derived the following mathematical consequences:

\begin{enumerate}
\item The refreshment rate is exponential.
\item Given activations $(a_{1}(t_{0}),\ldots , a_{n}(t_{0}))$ at time $t_{0}$, and providing that no new item is presented after $t_{0}$, the product of the activations at time $t>t_{0}$ does not depend on how the refreshment time is distributed across the memory items.
\item The refreshment rate $r$, decay rate $d$, and simple span $k$---which is the maximal number of items that can be maintained in memory in a simple span task---are linked by the straightforward relation $$k=\left \lfloor{1+\frac{r}{d}} \right \rfloor .$$
\end{enumerate}

%%%%%%%%%%%%%%
%%%%%%%%%%%%%%
\section{Variants of the TBRS model}
\label{sec:variants}

A given task defined by a function $T$ leads to a time-dependent focus vector $(\varphi_{1}(t),\ldots ,\varphi_{n}(t))$.
The TBRS assumptions require that
\begin{itemize}
\item $\varphi_{i}(t)$ is undefined if item $i$ has not yet been presented at $t$,
\item $\varphi_{i}(t)=0$ (or is undefined) if $T(t)=-1$,
\item $\varphi_{i}(t)=2$ (and $\varphi_{j\neq i}=0$ or is undefined) if $T(t)=i$, and
\item $(\varphi_{1}(t),\ldots ,\varphi_{n}(t))$ has exactly one component equal to 1, and all the others are equal to 0 or undefined, if $T(t)=0$.
\end{itemize}

The last point expresses that whenever no item is being presented, and when no concurrent task is required, attention is focused on one of the to-be-remembered items. However, it does not predict how spare time is distributed among items. We will now define six variants of the TBRS2 model based on how the spare time period is organized to deal with the memory items.

%%%%%%%%%%%%%%
\subsection{Steady vs.~threshold}
A first distinction can be made regarding how long an item is refreshed when attention is focused on it. A variant of the TBRS2 that we will call \emph{steady} posits that the refreshment duration is a fixed value (for instance, $d=0.2$ s).
Another variant (the \emph{threshold} model) posits that whenever attentional focus switches to a new item, it does so until activation of this particular item reaches a threshold $w$ (unless attentional focus is directed away by a concurrent task).

Figure \ref{fig:type} shows examples of predicted activation dynamics in the case of a simple span task for the steady and threshold variants. In these examples, the steady variant predicts more variability in the final probability of recall than the threshold variant if the number of items (here 3) is below the simple span. However, it predicts greater variability if the number of items (here 4) is above the simple span.

\begin{figure}[htbp]
\begin{center}
\includegraphics[width=\textwidth]{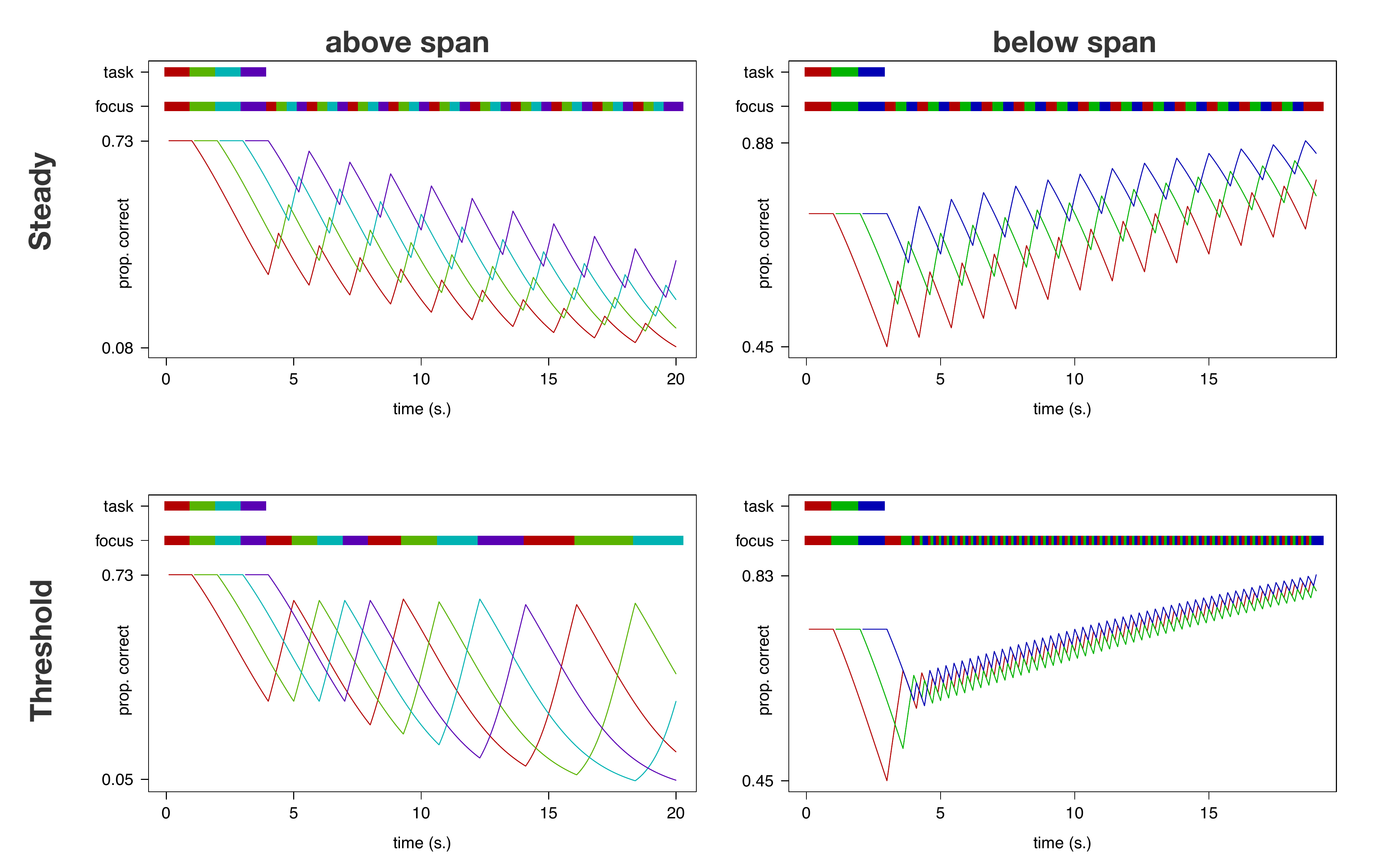}
\caption{Variants of TBRS2 predictions for a simple span task. We set $d=0.6$ and $r=1.4$. In the steady variant, refreshment of an item lasts 0.4 s, whereas in the threshold variant, refreshment stops when activation reaches a cut-off point, or after 0.1 s (i.e., the refreshment lasts 0.1 s if the activation is already above the threshold at the beginning of the refreshment period).}
\label{fig:type}
\end{center}
\end{figure}

%%%%%%%%%%%%%%
\subsection{First, next, or lowest}

The TBRS2 model can also vary as a function of how interruptions due to concurrent tasks are handled. During ``spare time,'' items are being refreshed in a regular order: item 1, item 2, item 3, ..., item $n$, item 1..., but there are different ways to select the to-be-refreshed item after an interruption caused by the concurrent task. We will consider three simple variants. (1) In the \emph{first} variant, the first item $x_{1}$ is always refreshed first after the presentation of a distractor. (2) In the \emph{next} variant, the model keeps track of the last refreshed item and continues with the next one. For instance, if the concurrent task occurs when item 2 is refreshed, then item 3 will be refreshed after the interruption. Finally, (3) the \emph{lowest} variant predicts that the item with the lowest activation is refreshed first. This could correspond to a ``maximin'' strategy in which one tries to maximize the minimal activation.

Depending on how dual task interruptions are spread in time, these variants may lead to different predictions that could hardly be presented by a verbal theory. A few illustrative examples are shown in Figure \ref{fig:startType}. With a particular task (alternation of a dual task and spare time every 2 s), visual inspection reveals different predictions depending on the variant. The \emph{first} version predicts a primary effect, whereby the first item is more likely to be recalled. The \emph{next} version predicts a recency effect, whereby the last item is more likely to be recalled. Finally, the \emph{lowest} version predicts similar decreases in activation among items and no clear-cut order effect.

\begin{figure}[htbp]
\begin{center}
\includegraphics[width=\textwidth]{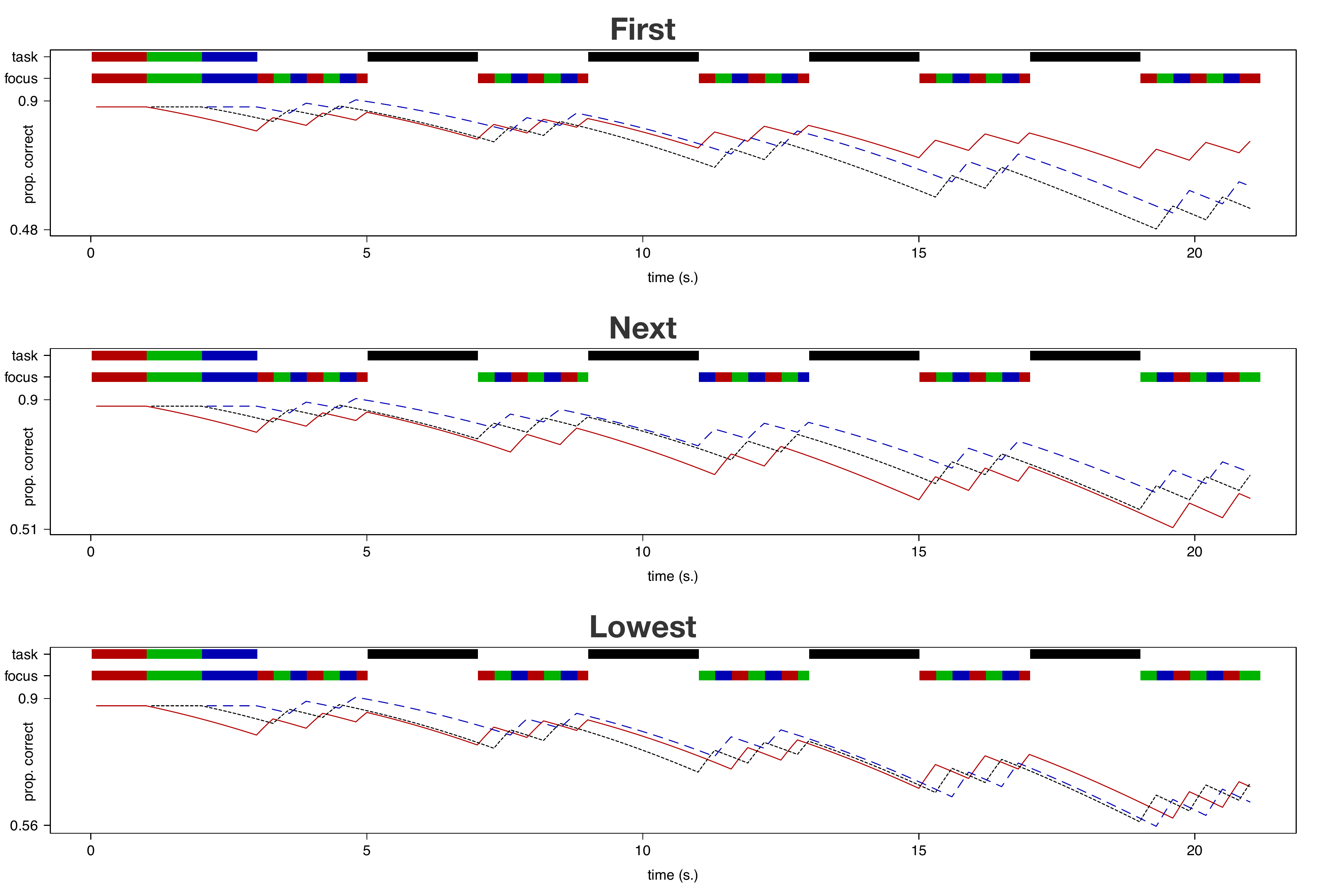}
\caption{Predictions of the steady TBRS with three variants, with alternating distractors and spare time (2 s each). Parameters are set as follows: $d=0.3, r=1$, duration$=0.3$ s, $\beta=\exp(2)$ (log-odds at presentation is set to 2). Here, we simplify a complex span task by putting the memory items at the beginning of the task.}
\label{fig:startType}
\end{center}
\end{figure}

%%%%%%%%%%%%%%
%%%%%%%%%%%%%%
\section{Parameter estimation and model comparison}

The TBRS is underdefined when it comes to two characteristics. First, it does not detail how long each item is refreshed. We suggest two simple models concerning this point: either a constant duration (steady) or a duration granting that activation reaches a cut-off point (threshold).

Second, TBRS does not detail how the to-be-refreshed item is determined after an interruption. We consider three versions of the model: (1) in the ``first'' version, refreshment is reset to the first item, (2) in the ``next'' version, the model keeps track of the previous item and goes to the next one. Finally, (3) in the ``lowest'' version, the less activated item is chosen as a new starting point.
Combining these possibilities, we can build six versions of the TBRS2. In the following, we analyze experimental data using these variants. Note that all six variants have the same number of parameters (4) listed below:

\begin{enumerate}
\item Decay rate $d$, expressed in points of log-odds by second,
\item Refreshment rate $r$, also expressed in points of log-odds by second,
\item Baseline $\beta$, which is the activation of an item when presented, and
\item Duration (of the focus of a particular item) or threshold $w$.
\end{enumerate}

In this last section, we use empirical data to illustrate how the formal framework of TBRS2 can be used to compare models, estimate parameters, and gauge TBRS assumptions.

%%%%%%%%%%%%%%
\subsection{Method}

Thirty-two psychology students aged 18--29 years ($M = 20.83, SD =2.58$) were recruited to take part in the experiment for course credit. Each participant performed a series of complex memory span task trials. In contrast to previous experiments in which the dual task is regular, the present dual task was semi-randomly organized along the timeline to test the TBRS using the most diversified patterns of distraction.
 
\subsubsection{Procedure}
The stimuli were capital letters (B, F, H, J, K, L, P, Q, R, S, V, X) chosen for having few phonological similarities in French. The stimuli were displayed visually on a computer screen.  Each list was composed of a maximum of six letters that were drawn without replacement.
 
After each letter, a concurrent task required the participant  to press the space bar whenever a 1, 2, or 3 stimulus digit was displayed. The stimulus digits were drawn randomly from the 1--9 set. The concurrent task occurred during a free time duration that was randomly drawn between 1 and 5 s (``free time'' indicates only that the participant was not presented with the to-be-recalled stimulus letters, but does not imply that they were free enough to refresh the letters, as explained below). This free time was divided into 1000 ms time slots during which attention capture could occur. For each slot, there was a chance for a distractor to be presented.
 
 Each experimental session lasted approximately half an hour and included 60 separate stimulus lists. The 60 lists were built as follows: the length varied from two to six letters, and the  difficulty varied from easy to difficult. There were six different Difficulty conditions for the concurrent task, based on six different probabilities (0, 0.20, 0.40, 0.60, 0.80, 1) that one stimulus digit would be drawn during each slot of the free time period. Once a  probability was set for a list, it was applied to the entire list across the slots. This generated 5 (Length) $\times$ 6 (Difficulty) $= 30$ conditions, which we doubled to give each participant 60 lists to be recalled. For instance, if a 5 s free time duration, divided into five slots of 1000 ms each, was chosen between two letters for a given list, and if the  probability was set to 0.80, the probability that one digit could be displayed in each of the five time slots was 0.80. This could, for instance, generate a 90827 sequence (with the ``0'' symbol indicating a period of 1000 ms without any distractor). The letters and the to-be-captured digits (1, 2, or 3) were always followed by an ``empty'' slot to avoid building sequences that would be too cluttered. An example of a sequence is K01050V02087X000Q980. The participants could enter their response by clicking on a visual keyboard of 3 $\times$ 4 letters. The letters were always associated with the same positions on the keyboard across trials. The letter disappeared after being clicked, so the participants were not able to correct their response. The subjects were instructed to recall the letters in order, if possible. After the participants validated their answer with the space bar, a feedback screen indicated whether the recall was correct (i.e., item memory and order memory both correct). Then, a screen with a GO window waited until the user moved on to the following list by pressing the space bar again. The participants were then presented with the next list, which followed a fixation cross lasting 2 s.
 
The task began after a short warmup including 18 progressive conditions. For the warmup only, a rapid green light appeared on the screen under the digit location whenever a digit was correctly captured. Similarly, a rapid red light appeared whenever the space bar was pressed in error (false alarm). After the warmup, the experimenter checked whether both tasks (memory task and concurrent task) were correctly performed during the warmup before running the actual experiment. Omissions, false alarms, and recall performance were scrutinized by the experimenter in order to give the best advice to the participant for the following experiment (for instance, being less impulsive on the concurrent task or being more attentive to the concurrent/memory task).

%%%%%%%%%%%%%%
\subsection{Results}
For this first simulation, we chose to analyze the data without taking the order or success of the concurrent task into account; a letter item was considered correctly recalled whenever it appeared in the response. For each participant and each variant of the TBRS2, a nonlinear minimization of $-LL$ (where $LL$ is the log-likelihood of the observed data) based on Newton's algorithm was performed. The constraints imposed on $d$ and $r$ were that $r>0$ and $2<r/d<11$. The results concerning the log-likelihood are displayed in Table \ref{tab:mll}.

To get a baseline, we computed the log-likelihood of a dummy model assigning equal and constant probability of recall to every item. For instance, if a participant recalled 95\% of the 240 memoranda on the whole, the dummy model predicted a probability of recall equal to 0.95 for every item in every trial. Note that the dummy model, although simplistic, gives an almost perfect fit for a proportion of correct recall nearing 1.

As shown in Table \ref{tab:mll}, the proportion of correct recall is high, ranging from 0.79 to 0.99. The TBRS2 clearly fits the data better than the dummy model in terms of $LL$. However, the TBRS2 has four free parameters, whereas the dummy model has only one free parameter. We thus used the Akaike information criterion (AIC) to compare the models. The results are given in Figure \ref{fig:mll}, in which the area below the dotted horizontal line corresponds to cases for which the TBRS2 yields poorer fit than the dummy model in terms of the AIC.

From the estimated parameters $d$ and $r$, we could derive an estimate of the simple span using Theorem \ref{theo:span}. As a result of constraints imposed on $r/d$, this estimated simple span was bound to lie between 3 and 12. The mean estimated simple span was 7.09 ($SD=1.53$), with a median equal to 8.

% latex table generated in R 3.2.2 by xtable 1.7-4 package
% Sun Nov 22 09:55:03 2015
\begin{table}[ht]
\caption{Participants' maximum log-likelihood table, sorted by increasing proportion of correct recall (column 1). The second column indicates the absolute value of the maximum log-likelihood of the dummy model. Columns 3 to 8 indicate the maximum log-likelihood difference between the dummy model and a variant of TBRS2. Positive values indicates that a variant of TBRS2 fits the data better than the dummy model. SF: Steady First, SN: Steady Next, SL: Steady Lowest, TF: Threshold First, TN: Threshold Next, TL: Threshold Lowest.}
\centering
\begin{tabular}{lrrrrrrr}
  \hline
Correct & Dummy & SF & SN & SL & TF & TN & TL \\ 
  \hline
0.79 & 124.14 & 13.90 & 14.36 & 13.81 & 11.82 & 11.88 & 11.46 \\ 
  0.80 & 118.70 & 11.50 & 11.35 & 13.82 & 12.53 & 12.43 & 16.57 \\ 
  0.81 & 115.82 & 9.41 & 9.41 & 10.60 & 11.25 & 12.61 & 13.15 \\ 
  0.82 & 117.65 & 2.12 & 2.11 & 2.21 & 0.04 & 2.11 & 3.26 \\ 
  0.85 & 101.45 & 10.58 & 9.46 & 9.27 & 13.63 & 15.56 & 13.33 \\ 
  0.88 & 90.42 & 15.03 & 14.14 & 13.94 & 13.61 & 14.60 & 11.95 \\ 
  0.88 & 86.46 & 8.98 & 9.29 & 7.98 & 9.36 & 12.16 & 9.58 \\ 
  0.91 & 71.21 & 8.87 & 9.14 & 8.74 & 8.40 & 8.57 & 7.75 \\ 
  0.92 & 66.42 & 3.84 & 3.94 & 2.66 & 3.02 & 3.11 & 3.11 \\ 
  0.93 & 61.39 & 7.70 & 7.33 & 6.16 & 5.86 & 6.52 & 3.97 \\ 
  0.93 & 58.78 & 1.77 & 1.40 & 1.47 & 1.54 & 1.49 & 1.57 \\ 
  0.94 & 56.11 & 2.14 & 2.18 & 1.71 & 2.21 & 3.24 & 2.88 \\ 
  0.94 & 56.11 & 5.86 & 7.14 & 6.17 & 5.98 & 6.41 & 6.08 \\ 
  0.94 & 56.24 & 0.46 & 0.47 & 0.67 & 0.38 & 0.38 & 1.20 \\ 
  0.95 & 50.55 & 1.16 & 1.14 & 0.95 & 1.42 & 1.44 & 0.87 \\ 
  0.95 & 50.55 & 3.65 & 3.48 & 2.22 & 3.30 & 3.72 & 2.92 \\ 
  0.95 & 44.65 & 12.57 & 12.70 & 14.41 & 13.69 & 13.82 & 14.92 \\ 
  0.96 & 41.57 & 4.79 & 3.83 & 2.59 & 5.28 & 5.18 & 3.26 \\ 
  0.96 & 38.38 & 1.43 & 1.44 & 2.25 & 1.75 & 1.91 & 2.55 \\ 
  0.96 & 38.38 & 0.78 & 0.68 & 0.77 & 1.48 & 1.19 & 0.15 \\ 
  0.96 & 38.38 & 0.67 & 0.70 & 0.58 & 1.21 & 1.75 & 0.04 \\ 
  0.97 & 31.64 & 4.54 & 4.61 & 4.27 & 4.78 & 4.37 & 5.43 \\ 
  0.97 & 28.06 & 2.29 & 2.38 & 1.90 & 2.09 & 1.89 & 2.21 \\ 
  0.98 & 24.30 & 1.52 & 1.22 & 0.99 & 0.93 & 1.36 & 0.00 \\ 
  0.98 & 20.34 & 3.23 & 3.88 & 3.05 & 4.61 & 2.58 & 2.34 \\ 
  0.98 & 20.34 & 0.37 & 0.40 & 0.60 & 0.31 & 0.40 & 0.84 \\ 
  0.98 & 20.34 & 3.45 & 3.00 & 3.09 & 3.62 & 3.58 & 2.48 \\ 
  0.99 & 16.13 & 0.69 & 0.70 & 0.57 & 0.36 & 0.37 & 0.42 \\ 
  0.99 & 16.13 & 0.35 & 0.36 & 0.32 & 0.56 & 0.60 & 0.13 \\ 
  0.99 & 11.57 & 0.67 & 0.68 & 0.78 & 0.51 & 0.51 & 0.74 \\ 
  0.99 & 11.57 & 0.23 & 0.23 & 0.19 & 0.20 & 0.20 & 0.23 \\ 
  0.99 & 11.57 & 0.39 & 0.48 & 0.38 & 0.20 & 0.35 & 0.36 \\ 
   \hline
\end{tabular}
\label{tab:mll}
\end{table}

\begin{figure}[htbp]
\begin{center}
\includegraphics[width=\textwidth]{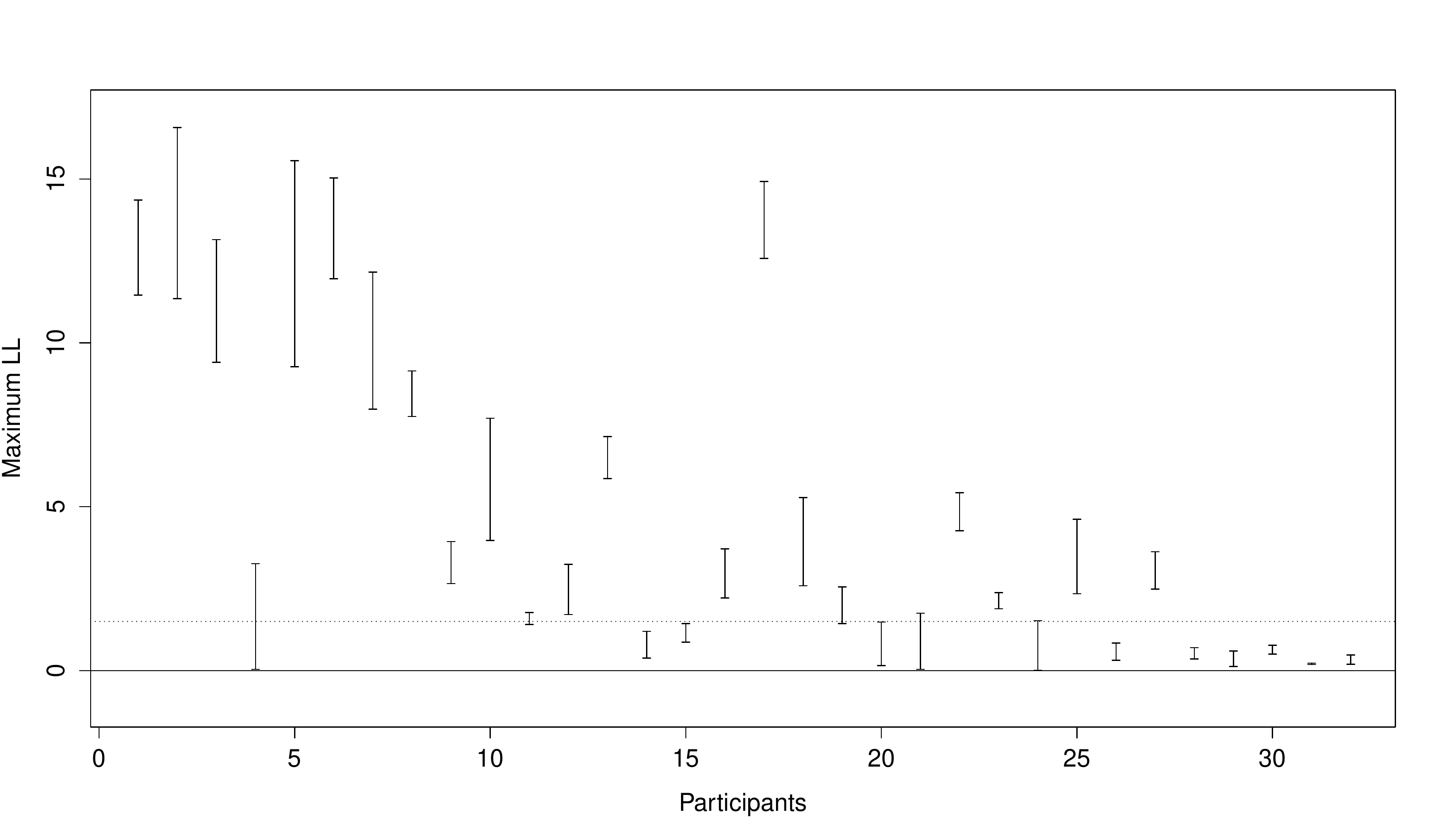}
\caption{Maximum log-likelihood (LL) by participant, sorted by increasing proportion of correct recall. The $y$ axis displays the maximum log-likelihood difference between the TBRS2 and dummy model. Each segment runs from the minimum to the maximum value across the variants. Positive values indicate that TBRS2 fits the data better than the dummy model. The bottom area between the solid and dotted lines corresponds to cases in which, although the LL is greater for the TBRS2, the Akaike information criterion is lower.}
\label{fig:mll}
\end{center}
\end{figure}

%%%%%%%%%%%%%%
\subsection{Discussion}

Comparing the TBRS2 with a dummy model yields apparently mixed results. Because the observed proportion of correct recall often approaches 1 (with 9 participants above 98\% of correct recall), one should take into account the observed proportion of correct recall. As illustrated in Figure \ref{fig:mll}, the TBRS2 variants gave a better fit (AIC) than a dummy model whenever a participant correctly recalled less than 94\%. Keeping that in mind, the data are clearly in favor of the TBRS2, as compared with the dummy model.

Comparisons of non-embedded models based on the AIC are always subjective, but some authors have suggested that a difference of 4 to 7 (i.e., a difference of 2 to 3.5 in terms of $LL$ for the TBRS2) roughly corresponds to a 95\% confidence interval \cite{burnham2002model}. Using this rule of thumb, we found no strong evidence in favor of any variant of the TBRS2 against another. There was no overall best variant across the participants, partly because several variants fitted different participants without apparent regularity. We defer to future research the task of comparing variants of the TBRS2, as it is not the main goal of the current study.

Using Theorem \ref{theo:span}, we derived participant-wise estimates of the simple span and found that, although we imposed only loose constraints on this simple span, the resulting estimates are in line with previous research suggesting a simple span of 5 to 9. In fact, only two participants did not fall within this range. This is an argument in favor of the TBRS2, and by extension an argument in favor of the TBRS theory.

%%%%%%%%%%%%%%
%%%%%%%%%%%%%%
\section{Conclusion}

We built the first detailed mathematical transcription of the TBRS assumptions that adds no characteristic not already addressed by the original description \cite{barrouillet2007time}, making as few decisions as possible and using as few parameters as possible. In comparison to the only other computational implementation (TBRS*) \cite{oberauer2011tbrs}, our TBRS2 model does not account for features such as order encoding: Although it does predict order effects on correct recall, it does not describe how the order of the items is encoded. On the one hand, the TBRS2 is thus less rich than the TBRS*. On the other hand, the TBRS2 is simpler and more transparent. Thanks to this transparency, we were able to prove several theorems mathematically following from the TBRS assumptions. For instance, the decay and refreshment functions are tightly related under the cognitive load assumption. Another striking theoretical result is that the simple span can be computed from the decay and refreshment rates.
These results can now be used to test the TBRS theory at a degree of precision probably never reached before. In an illustrative experiment, we estimated the TBRS2 free parameters $d$ and $r$ and derived a simple span estimate. We found plausible results in favor of the TBRS2 and therefore in favor of the TBRS theory.

%%%%%%%%%%%%%%
%%%%%%%%%%%%%%
\section*{References}

\bibliography{TBRSbib}
\bibliographystyle{plain}

%%%%%%%%%%%%%%
%%%%%%%%%%%%%%
\end{document}